%% file: main.tex
\documentclass{llncs}

\usepackage{makeidx}
\usepackage{graphicx}
\usepackage{amssymb}
\usepackage{amsmath}
\usepackage{mathtools}
\usepackage[utf8]{inputenc}
\usepackage{nicefrac}
\usepackage{enumerate}
\usepackage{enumitem}
\usepackage[linesnumbered,vlined]{algorithm2e}
\usepackage{hyperref}
\usepackage{cleveref}

\input{macros.tex}

\begin{document}

\DeclarePairedDelimiterX\Sett[2]{\lbrace}{\rbrace}%
{ #1 \,\delimsize|\, #2 }
\title{The Input/Output Complexity of Sparse Matrix Multiplication\thanks{This work is supported by the Danish National Research Foundation under the Sapere Aude program.}}

\author{Rasmus Pagh \and Morten Stöckel}

\institute{IT University of Copenhagen, \\ \email{ \{pagh,mstc\}@itu.dk}}

\maketitle

\begin{abstract}
We consider the problem of multiplying sparse matrices (over a semiring) where the number of non-zero entries is larger than main memory.
In the classical paper of Hong and Kung (STOC '81) it was shown that to compute a product of dense $U \times U$ matrices, $\Theta \left( U^3 / (B \sqrt{M}) \right)$ I/Os are necessary and sufficient in the I/O model with internal memory size $M$ and memory block size $B$.

In this paper we generalize the upper and lower bounds of Hong and Kung to the sparse case.
Our bounds depend of the number $N = \nnz(A)+\nnz(C)$ of nonzero entries in $A$ and $C$, as well as the number $Z = \nnz(AC)$ of nonzero entries in $AC$.

We show that $AC$ can be computed using $\tilde{O} \left( \tfrac{N}{B} \min\left(\sqrt{\tfrac{Z}{M}},\tfrac{N}{M}\right) \right)$ I/Os, with high probability.
This is tight (up to polylogarithmic factors) when only semiring operations are allowed, even for dense rectangular matrices:
We show a lower bound of $\Omega \left( \tfrac{N}{B} \min\left(\sqrt{\tfrac{Z}{M}},\tfrac{N}{M}\right) \right)$ I/Os.

While our lower bound uses fairly standard techniques, the upper bound makes use of ``compressed matrix multiplication'' sketches, which is new in the context of I/O-efficient algorithms, and a new matrix product size estimation technique that avoids the ``no cancellation'' assumption.
\end{abstract}

\input{introduction.tex}
\input{sizeest.tex}
\input{awareupper.tex}
\input{lowerbound.tex}

\bibliographystyle{abbrv}
\bibliography{interbib}

\end{document}

%% file: macros.tex
\newtheorem{fact}[theorem]{Fact}

\crefname{fact}{Fact}{Facts}
\DeclarePairedDelimiterX\Set[2]{\lbrace}{\rbrace}%
 { #1 \,\delimsize|\, #2 }

\def\FullBox{\hbox{\vrule width 8pt height 8pt depth 0pt}}

\def\qed{\ifmmode\qquad\FullBox\else{\unskip\nobreak\hfil
\penalty50\hskip1em\null\nobreak\hfil\FullBox
\parfillskip=0pt\finalhyphendemerits=0\endgraf}\fi}

\def\qedsketch{\ifmmode\Box\else{\unskip\nobreak\hfil
\penalty50\hskip1em\null\nobreak\hfil$\Box$
\parfillskip=0pt\finalhyphendemerits=0\endgraf}\fi}

\DeclareMathOperator{\emit}{\mathtt{emit}}
\DeclareMathOperator{\nnz}{\mathtt{nnz}}

\newcommand{\sort}{\textrm{sort}}

%% file: introduction.tex
%!TEX root = main.tex

\section{Introduction}
In this paper we consider the fundamental problem of multiplying matrices that are \emph{sparse}, that is, the number of nonzero entries in the input matrices (but not necessarily the output matrix) is much smaller than the number of entries.
Matrix multiplication is a fundamental operation in computer science and mathematics, due to the wide range of applications and reductions to it --- e.g.~computing the determinant and inverse of a matrix, or Gaussian elimination.
Matrix multiplication has also seen lots of use in non-obvious applications such as bioinformatics~\cite{vandongen00}, computing matchings~\cite{Rabin:1989:MMG:72302.72307,mvv89} and algebraic reasoning about graphs, e.g. cycle counting~\cite{Alon:1995:COL:210332.210337,aloncount}.

Matrix multiplication in the general case has been widely applied and studied in a pure math context for decades.
In an algorithmic context matrix multiplication is known to be computable using $O(n^{\omega})$ ring operations, for some constant $\omega$ between 2 and 3.
The first improvement over the trivial cubic algorithm was achieved in $1969$ in the seminal work of Strassen~\cite{strassen} showing  $\omega \leq \log_2 7$ and most recently Vassilevska Williams~\cite{Williams:2012:MMF:2213977.2214056} improved this to $\omega < 2.373$.

Matrix multiplication over a semiring, where additive inverses cannot be used, is better understood.
In the I/O model introduced by Aggarwal and Vitter~\cite{AV1988} the optimal matrix multiplication algorithm for the dense case already existed (see \Cref{sec:rel}) and since then sparse-dense and sparse-sparse combinations of vector and matrix products have been studied, e.g. in~\cite{brodal:spmv,greiner2010complexity,Pagh:2013:CMM:2493252.2493254}.

The main contribution of this paper is a tight bound for matrix multiplication over a semiring in terms of the number of nonzero entries in the input and output matrices, generalizing the classical result of Hong and Kung on dense matrices~\cite{Jia-Wei:1981:ICR:800076.802486} to the sparse case.

\subsection{Preliminaries}\label{sec:pre}
Let $A\in R^{U \times U}$ and $C\in R^{U \times U}$ be matrices of $U$ rows and $U$ columns and let every entry $[A]_{i,j}, [C]_{i',j'} \in R$ for semiring $R$. Further for matrix $A$ let $A_{i*}$ denote row $i$ of $A$ and let $A_{* j}$ denote column $j$ of $A$.
The matrix product $AC$, where each entry $[AC]_{i,j}$, $i,j \in [U]$ is given as $[AC]_{i,j} = \sum_{k} [A]_{i,k} [C]_{k,j}$.
A nonzero term $[A]_{i,k} [C]_{k,j}$ is referred to as an \emph{elementary product}.
We say that there is \emph{no cancellation} of terms when  $[AC]_{i,j} = 0$ implies that $[A]_{i,k} [C]_{k,j} = 0$ for all $k$.
For \emph{sparse} semiring matrix multiplication, the number of entry pairs with nonzero product measures the number of operations performed up to a constant factor assuming optimal representation of the matrices. Specifically, let $\sum_{k=1}^n |\Sett{j}{[A]_{j,k} \neq 0}||\Sett{i}{[C]_{k,i} \neq 0}|$ be the number of such nonzero pairs of matrix entries. Finally let $\nnz(A) = |\Sett{i,j}{[A]_{i,j} \neq 0}|$ denote the number of nonzero entries of matrix $A$.
When no explicit base is stated, logarithms in this paper are base $2$.

\medskip

{\bf External memory model.} This model of computation~\cite{AV1988} is an abstraction of a two-level memory hierachy: We have an internal memory holding $M$ data items (``words'') and a disk of infinite size holding the remaining data. Transfers between internal memory and disk happen in blocks of $B$ words, and a word must be in internal memory to be manipulated. The cost of an algorithm in this model is the number of block transfers (I/Os) done by the algorithm. We will use $\sort(n) = O((n/B) \log_{M/B}(n/B))$ as shorthand for the sorting complexity of $n$ data items in the external memory model and $\tilde{O}(\cdot)$-notation to suppress polylogarithmic factor in input size $N$ and matrix dimension $U$.

We assume that a word is big enough to hold a matrix element from a semiring as well as the matrix coordiantes of that element, i.e., a block holds $B$ matrix elements. We restrict attention to algorithms that work with semiring elements as an abstract type, and can only copy them, and combine them using semiring operations. We refer to this restriction as the {\em semiring I/O model}. Our upper bound uses a slight extension of this model in which equality check is allowed, which allows us to take advantage of {\em cancellations}, i.e., inner products in the matrix product that are zero in spite of nonzero elementary products.

\medskip

{\bf The problem we solve.}
Given matrices $A\in R^{U \times U}$ and $C\in R^{U \times U}$ containing $\nnz(A)$ and $\nnz(C)$ non-zero semiring elements, respectively, we wish to output a sparse representation of the matrix product $AC$ in the external memory model.
We are dealing with sparse matrices represented as a list of tuples of the form $(i,j,[A]_{ij})$, where $[A]_ij \in R$ is a (nonzero) matrix entry.
To produce output we must call a function $\emit(e)$ for every nonzero entry $(i,j,(AC)_ij)$ of $AC$.
We only allow $\emit(\cdot)$ to be called once on each output element, but impose no particular order on the sequence of outputs.

We note that the algorithm could be altered to write the entire output before termination by, instead of calling $\emit(\cdot)$, simply writing the output element to a disk buffer, outputting all $\nnz(AC)$ elements using $O(\nnz(AC)/B)$ additional I/Os.
However, in some applications such as database systems (see~\cite{Amossen:2009:FJS:1514894.1514909}) there may not be a need to materialize the matrix product on disk, so we prefer the more general method of generating output.

\subsection{Related work}\label{sec:rel}
The external memory model was introduced by Aggarwal and Vitter in their seminal paper~\cite{AV1988}, where they provide tight bounds for a collection of central problems.

An I/O-optimal matrix multiplication algorithm for dense semiring matrices was achieved by Hong and Kung~\cite{Jia-Wei:1981:ICR:800076.802486}:
Group the matrices into $k\sqrt{M} \times k\sqrt{M}$ submatrices where constant $k$ is picked such that three $\sqrt{M} \times \sqrt{M}$ matrices fit into internal memory.
This reduces the problem to $O ( (U^3 / \sqrt{M})^3 )$ matrix products that fit in main memory, costing $O(M/B)$ I/Os each, and hence $O( U^3 / B \sqrt{M} )$ in total~\cite{BRICS2002}.
Hong and Kung also provided a tight lower bound $\Omega( U^3 / B \sqrt{M} )$ that holds for algorithms that work over a semiring.
(It does not apply to algorithms that make use of subtraction, such as fast matrix multiplication methods, for which the blocking method described above yields an I/O complexity of $U^\omega / (M^{\omega/2-1} B)$ I/Os.)
\begin{sloppypar}
For sparse matrix multiplication the previously best upper bound~\cite{Amossen:2009:FJS:1514894.1514909}, shown for Boolean matrix products but claimed for any semiring, is $\tilde{O}(N \sqrt{\nnz(AC)} / B M^{1/8})$.
\end{sloppypar}
It seems that this bound requires ``no cancellation of terms'' (or more specifically, the output sensitivity is with respect to the number of output entries that have a nonzero elementary product).
 Our new upper bound of this paper improves upon this: The Monte Carlo algorithm of \Cref{thm:main1} has strictly lower I/O complexity for the entire parameter space and makes no assumptions about cancellation.

An important subroutine in our algorithm is dense-vector sparse matrix multiplication:  For a vector $y$ and sparse matrix $S$ we can compute their product using optimal $\tilde{O}((\nnz(S)+\nnz(y))/B)$ I/Os~\cite{brodal:spmv} - this holds for arbitrary layouts of the vector and matrix on disk.

Our algorithm has an interesting similarity to Williams and Yu's recent output sensitive matrix multiplication algorithm~\cite[Section 6]{doi:10.1137/1.9781611973402.135}.
Their algorithm works by splitting the matrix product into 4 submatrices of equal dimension, running a randomized test to determine which of these subproblems contain a nonzero entry.
Recursing on the non-zero submatrices, they arrive at an output sensitive algorithm.
We perform a similar recursion, but the splitting is computed differently in order to recurse in a balanced manner, such that each subproblem at a given level of the recursion outputs approximately the same number of entries in the matrix product.

Size estimation of the number of nonzeros in matrix products was used by Cohen~\cite{cohen:1998,Cohen:1994:EST:1398518.1399012} to compute the order of multiplying several matrices to minimize the total number of operations.
For constant error probability this algorithm uses $O( \varepsilon^{-2} N )$ operations in the RAM model to perform the size estimation.
For $\varepsilon > 4 / \nnz(AC)^{1/4}$ Amossen et al~\cite{Amossen:2010:BSE:1886521.1886554} improved the running time to be expected $O(N)$ in the RAM model and expected $O(\sort(N))$ in the I/O model.
Contrary to the approaches of~\cite{Amossen:2010:BSE:1886521.1886554,cohen:1998,Cohen:1994:EST:1398518.1399012} our new size estimation algorithm presented in \Cref{sec:sizest} is able to deal with cancellation of terms, and it uses $\tilde{O} ( \varepsilon^{-3} N / B )$ I/Os.
Informally, the main idea of our size estimation algorithm is to multiply a sequence of vectors $x$ with certain properties onto $AC$ but in the computationally inexpensive order $(xA)B$, in order to produce linear sketches of the rows (columns) of $AC$. %This is often seen in variants of iterative algebra, see e.g. \cite{greiner2010complexity}.

\subsection{New results}

We present a new upper bound in the I/O model for sparse matrix multiplication over semirings. Our I/O complexity is at least a factor of roughly $M^{3/8}$ better than that of~\cite{Amossen:2009:FJS:1514894.1514909}.
We show the following theorem:
\begin{theorem}\label{thm:main1}
Let $A\in R^{U \times U}$ and $C\in R^{U \times U}$ be matrices with entries from a semiring $R$, and let $N = \nnz(A)+\nnz(C)$, $Z = \nnz(AC)$. There exist algorithms (a) and (b) such that:
\begin{enumerate}[label=(\alph*)]
\item emits the set of nonzero entries of $AC$ with probability at least $1 - 1/U$, using $\tilde{O} \left( N \sqrt{Z} / (B \sqrt{M}) \right)$ I/Os.
\item emits the set of nonzero entries of $AC$, and uses $O \left( N^2 / (MB) \right)$ I/Os.
\end{enumerate}
For every $A$ and $C$, using $\tilde{O} \left( N / B \right)$ I/Os we can determine with probability at least $1-1/U$ if one of the two I/O bounds is significantly lower, i.e., distinguish between $N\sqrt{Z} / (B \sqrt{M}) >  2 N^2 / (MB)$ and $2 N\sqrt{Z} / (B \sqrt{M}) < N^2 / (MB)$.
\end{theorem}
The above theorem makes no assumptions about cancellation of terms. In particular, $\nnz(AC)$ can be smaller than the number of output entries that have nonzero elementary products.

Our second main contribution is a new lower bound on sparse matrix multiplication in the semiring I/O model.
\begin{theorem}\label{thm:lower}
	For all positive integers $N$ and $Z<N^2$ there exist matrices $A$ and $C$ with $\nnz(A),\nnz(C) \leq N$, $\nnz(AC) \leq Z$, such that computing $AC$ in the semiring I/O model requires $\Omega \left(\min \left( \frac{N^2}{MB},\frac{N\sqrt{Z}}{\sqrt{M} B} \right) \right)$ I/Os.
\end{theorem}
Since we can determine and run the algorithm satisfying the minimum complexity of the lower bound, our bounds are tight.

\medskip

{\bf Paper structure}. \Cref{sec:sizest} describes a new size estimation algorithm which we will use as a subprocedure for our sparse matrix multiplication algorithm. The new size estimation algorithm may be of independent interest since to the knowledge of the authors there are no published size estimation procedures that handle cancellation of terms. \Cref{sec:upper} first describes a simple output insensitive algorithm in \Cref{sec:outalg}, algorithm (b) of \Cref{thm:main1}. Then we describe how algorithm (a) of \Cref{thm:main1} works: In \Cref{sec:color} we describe how to divide the sparse matrix product into small enough subproblems (with respect to output size), and \Cref{sec:compressedmm} desribes how a version of Pagh's ``compressed matrix multiplication'' algorithm yields an I/O efficient algorithm for  subproblems with a small output. 
Finally in \Cref{sec:lowerbound} we show the new tight lower bound of \Cref{thm:lower}.

%% file: sizeest.tex
%!TEX root = main.tex

\section{Matrix output size estimation} \label{sec:sizest}
We present a method to estimate column/row sizes of a matrix product $AC$, represented as a sparse matrix. In particular, for a column $C_{*k}$ (or analogously row $A_{k*}$) we are interested in estimating the number of nonzeros $\nnz(A[C]_{*k})$ ($\nnz([A]_{k*}B))$. We note that there are no assumptions about (absence of) cancellation of terms in the following.
We show the existence of the following algorithm.

\begin{lemma}\label{lem:colest}
 Let $A\in R^{U \times U}$ and $C\in R^{U \times U}$ be matrices with entries from  semiring $R$, $N = \nnz(A)+\nnz(C)$ and let $0 < \varepsilon, \delta \leq 1$.
 We can compute estimates $z_1,\ldots,z_k$ using $\tilde{O}(\varepsilon^{-3}N/B)$ I/Os and $O(\varepsilon^{-3} N \log(U/\delta) \log U)$ RAM operations such that with probability at least $1 - \delta$ it holds that $(1-\varepsilon)\nnz([AC]_{*k}) \leq z_k \leq (1+\varepsilon)\nnz([AC]_{*k})$ for all $1 \leq k \leq U$.
\end{lemma}
We note that \Cref{lem:colest} by symmetry can give the same guarantees for rows of the matrix product, which is done analogously by applying the algorithm to the product $(AC)^T = C^TA^T$.
Further, from \Cref{lem:colest} we have, following from combining of all column estimates, an estimate of $\nnz(AC)$.
\begin{corollary}\label{lem:nnzest}
 Let $A\in R^{U \times U}$ and $C\in R^{U \times U}$ be matrices with entries from  semiring $R$, $N = \nnz(A)+\nnz(C)$ and let $0 < \varepsilon, \delta \leq 1$.
We can compute $\hat{Z}$ in $\tilde{O}(\varepsilon^{-3}N/B)$ I/Os and $O(\varepsilon^{-3} N \log(U/\delta) \log U)$ RAM operations such that with probability at least $1 - \delta$ it holds that $(1-\varepsilon)\nnz(AC) \leq \hat{Z} \leq (1+\varepsilon)\nnz(AC)$.
\end{corollary}
At a high level, the algorithm is similar in spirit to Cohen~\cite{cohen:1998,Cohen:1994:EST:1398518.1399012}, but uses linear $F_0$ sketches (see e.g.~\cite{12823057,Kane:2010:OAD:1807085.1807094}) that serve the purpose of capturing cancellation of terms.

We will make use of a well-known $F_0$-sketching method~\cite{12823057,mcgregorbook}, where $F_0(f)$ denotes the number of non-zero entries in a vector $f$.
Let $S$ be a data stream of items of the form $((i,j),r)$, where $(i,j)\in U\times U$ and $r\in R$.
The stream defines a vector indexed by $U\times U$ (which can also be thought of as a matrix), where entry $(i,j)$ is the sum of all ring elements $r$ that occurred with index $(i,j)$ in the stream. 

For a matrix $S \in R^{U \times U}$ the number of distinct indices
is the sum of distinct indices over all column vectors $F_0(S) = \sum_{i \in [U]} F_0(S_{i*})$.
One can compute in space $O( \varepsilon^{-3} \log n \log \delta^{-1} )$~\cite{12823057,mcgregorbook} a \emph{linear} sketch over $x$ that can output a number $\hat{z}$, where $(1-\varepsilon)F_0 \leq \hat{z} \leq (1+\varepsilon)F_0$ with constant probability.

\medskip

\paragraph{High-level algorithm description.} We compute a linear sketch $F$ followed by the matrix product $v = FAC$. From $v$ for a given $T$ we can distinguish between a column having more than $(1+\varepsilon)T$ and less than $(1-\varepsilon)T$ nonzero entries - we repeat this procedure for suitable values of $T$ to achieve the final estimate. We use the following distinguishability result:
\begin{fact}\label[fact]{fact:distsketch}(\cite{mcgregorbook}, Section 2.1)
There exists a projection matrix $M \in \{0,1\}^{n \times d}$ such that for each frequency vector $f \in R^{1 \times n}$ we can be estimate $F_0(f)$ from $fM$. In particular, for fixed $T' > 0$, $0 < \varepsilon', \delta' \leq 1$ with probability $1- \delta'$ we can distinguish the cases
$F_0(f) > (1+\varepsilon')T'$ and $F_0(f) < (1-\varepsilon')T'$ using space $d = O( \varepsilon'^{-2} \log \delta'^{-1})$.
\end{fact}
\noindent
We will apply this distinguishability sketch to the columns of the product $AC$, since $F_0(AC) > (1+\varepsilon)T$ implies $\nnz(AC) > (1+\varepsilon)T$ and analogously for the second case.
This follows trivially from the definition of $F_0$ and the number of nonzeroes in a matrix product. From \Cref{fact:distsketch} we have a sketch $F \in \{0,1\}^{d \times U}$ which multiplied with a matrix $S\in R^{U \times U}$ we can for the columns $[FS]_{*k}$ distinguish $\nnz(S_{*k}) > (1+\varepsilon)T$ from $\nnz(S_{*k}) < (1-\varepsilon)T$ with probability $1-\delta$.

\begin{proof}(\Cref{lem:nnzest})
Let $F \in \{0,1\}^{d \times U}$ be a $F_0$-distinguishability sketch as described in \Cref{fact:distsketch}.
To ensure that for every of the $U$ columns in $v = FAC$ we can distinguish the two cases with probability at least $1 - \delta$ it is sufficient to invoke the algorithm from \Cref{fact:distsketch} with $\delta' = \delta/U$. By the union bound over the error probabilities we have $\sum_{1 \leq i \leq U} \delta/U = \delta$. By linearity of $F$ we have that from $v$ we can for all columns $k \in[U]$ distinguish the cases $[AC]_{*k} < (1-\varepsilon)T$ and $[AC]_{*k} > (1+\varepsilon)T$.

Also by linearity, the order of operations in the computation of $v$ is $v = (FA)C$, hence the computation of $v$ can be seen as $2d$ dense-vector sparse-matrix multiplications.

For dense vector $y^{1 \times U}$ and sparse matrix $S\in R^{U \times U}$ we can compute $yS$ in $O( (\nnz(S)/B) \log_{M/B}(U/M) )$ I/Os~\cite{brodal:spmv}. Letting $N = \nnz(A) + \nnz(C)$ computing $v$ for a value of $T$ has I/O complexity
\begin{align}
O(2d (N/B) \log_{M/B} (U/M) &= O(\varepsilon^{-2} (N/B) \log_{M/B} (U/M) \log (U/\delta) \nonumber\\
&= \tilde{O}( \varepsilon^{-2} N/B)\label{eq:dist1}.
\end{align}
We note that for sparse matrices this bound is $\tilde{O}(\sort(U))$.
\noindent
Analogously, the number of RAM operations needed to compute $v$ for a specific $T$ is $O( \varepsilon^{-2} N \log(U / \delta))$.

Since for a given $T$ we can now using I/Os given in (\ref{eq:dist1}) distinguish $[AC]_{*k} < (1-\varepsilon)T$ and $[AC]_{*k} > (1+\varepsilon)T$ we simply repeat this procedure for $O(\varepsilon^{-1} \log U)$ values $T = 1, (1+\varepsilon), (1+\varepsilon)^2, \ldots, O(U)$, which yields a $1\pm \varepsilon$ estimate of the number of nonzeroes in each column, from which we have the desired estimate
of the total number of nonzeroes.
\qed
\end{proof}
We note that the algorithm of \Cref{lem:nnzest} can be obtained using any linear $F_0$ sketch in I/O complexity $O( \xi (N/B) \log_{M/B} (U/M))$, where $\xi$ is the space complexity of the sketch used.

%% file: awareupper.tex
%!TEX root = main.tex

\section{Cache-aware upper bound}\label{sec:upper}

As in the previous section let $A\in R^{U \times U}$ and $C\in R^{U \times U}$ be matrices with entries from a semiring $R$, and let $N = \nnz(A) + \nnz(C)$ be the input size.

\subsection{Output insensitive algorithm}\label{sec:outalg}

We first describe algorithm (a) of \Cref{thm:main1}, which is insensitive to the number of output entries $\nnz(AC)$.
It works as follows:
First put the entries of $C$ in column-major order by lexicographic sorting.
For every row $a_i$ of $A$ with more than $M/2$ nonzeros, compute the vector-matrix product $a_i C$ in time $\tilde{O}(N/B)$ using the algorithm of~\cite{brodal:spmv}.
There can be at most $2N/M$ such rows, so the total time spent on this is $\tilde{O}(N^2/(MB))$.
The remaining rows of $A$ are then gathered in groups with between $M/2$ and $M$ nonzero entries per group.
In a single scan of $C$ (using column-major order) we can compute the product of each such row with the matrix $C$.
The number of I/Os is $O(N/B)$ for each of the at most $2N/M$ groups, so the total complexity is $\tilde{O}(N^2/(MB))$.

\subsection{Monte Carlo algorithm overview}

We next describe algorithm (b) of \Cref{thm:main1}
The algorithm works by first performing a step of \emph{color coding}, the purpose of which is to split the matrix product into submatrices, each of which can be computed efficiently.
Roughly, the idea is to color the rows of $A$ and columns of $C$, forming submatrices $A_1,A_2,\dots$ and $C_1,C_2,\dots$ corresponding to each color, such that every matrix product $A_i C_j$ has roughly $M / \log U$ nonzero elements.
Then, a ``compressed'' matrix multiplication algorithm (described by \Cref{lem:compressedmm}) is used to compute every product $A_i C_j$ by a single scan over the matrices $A_i$ and $C_j$.
The number $c$ of colors needed to achieve this, to be specified later, depends on an estimate of $\nnz(AC)$, found using \Cref{lem:nnzest}.
It turns out to simplify the algorithm if we deviate slightly from the above coloring, by using different colorings of the rows of $C$ for each $A_i$.
That is, we will work with $c$ different decompositions $C^{(i)}_1,C^{(i)}_2,\dots$ of $C$, and compute products of the form $A_i C^{(i)}_j$.
Also, there might be rows of $A$ and columns of $C$ that we cannot color because they generate too many entries in the output.
However, it turns out that we can afford to handle such rows/columns in a direct way using vector-matrix multiplication.

\input{compressedmm.tex}

\subsection{Computing a balanced coloring}\label{sec:color}

We wish to assign every row $A_{k*}$ and column $C_{*k}$ a color from $[c]$. Let color set $S_i$ contain rows $A_{k*}$ that are assigned color $i$ and for such a color $i$ assigned to rows of of$A$ let color set $S^{(i)}_j$ contain columns $C_{*k}$ that are assigned color $j$.
Also, let $A|S_i$ be the input matrix $A$ restricted to contain only elements in rows from $S_i$ (and analogously for $C$ and $S^{(i)}_j$).

The goal of the coloring step is to assign the colors such that for every pair of color sets $(S_i, S^{(i)}_j)$, $1 \leq i,j \leq c$ it holds that $\nnz((A|S_i)(C|S^{(i)}_j)) < \gamma M / \log U$. This can be seen as coloring the rows of $A$ once and the columns of $C$ $c$ times.

\begin{lemma}\label{lem:color}
Let $A\in R^{U \times U}$ and $C\in R^{U \times U}$ be matrices with $N = \nnz(A)+\nnz(C)$ nonzero entries.\\
Using $\tilde{O}\left(\frac{N\sqrt{\nnz(AC)}}{B\sqrt{M}}\right)$ I/Os a coloring with $c = \sqrt{\frac{\nnz(AC) \log U}{M}} + O(1)$ colors can be computed that assigns a
color to rows of $A$ and for each such color $i$, assigns colors to columns of $C$ such that:
\begin{enumerate}
	\item For every $i,j \in [c]$ it holds that $\nnz \left( (A|S_i)(C|S^{(i)}_j) \right) < M/\log U$.
	\item Rows from $A$ and columns form $C$ that are not in some color sets $S_i$ and $S^{(i)}_j$ has had their nonzero output entries emitted.
\end{enumerate}
\end{lemma}
\begin{proof}
At a high level, the coloring will be computed by recursively splitting the matrix rows in two disjoint parts to form matrices $A_1$ and $A_2$ where $A_1$ contains the nonzeros from the first $t-1$ rows, for some $t$, and $A_2$ contains the nonzeros from the last $U-t$ rows.
Row number $t$, the ``splitting row'', will be removed from consideration by generating the corresponding part of the output using I/O-efficient vector-matrix multiplication.
We wish to choose $t$ such that:
\begin{enumerate}
	\item $\nnz(A_1 C) \in \left[ (1-\log^{-1} U)\nnz(AC)/2; (1+\log^{-1} U)\nnz(AC)/2 \right]$. \label{enu:split1}
	\item $\nnz(A_2 C) \in \left[ (1-\log^{-1} U)\nnz(AC)/2; (1+\log^{-1} U)\nnz(AC)/2 \right]$. \label{enu:split2}
\end{enumerate}
And after $ \log c + O(1)$ recursive levels of such splits, we will have $O(c)$ disjoint sets of rows from $A$.
For each such set we then compute disjoint column sets of $C$ in the same manner, and we argue below that this gives us subproblems with output size $\nnz(AC) / c^2 = M / \log U$, where each subproblem corresponds exactly to a pair of color sets as described above.

In order to compute the row number $t$ around which to perform the split, we invoke the estimation algorithm from \Cref{lem:nnzest} with $\varepsilon = \log^{-1} U$ such that for every row in $[AC]_{k*}$ we have access to an estimate $\hat{z}_k$ where it holds with probability at least $1 - U^{-l}$ (for fixed $l > 0$ chosen to get sufficiently low error probability):
\begin{equation}
 \hat{z}_k \in \left[ (1-\log^{-1} U)\nnz([AC]_{k*})/2; (1+\log^{-1} U)\nnz([AC]_{k*})/2 \right]. \label{eq:zest}
\end{equation}
In particular for any set of rows $r$ we have that
\begin{equation}
(1-\log^{-1} U)\nnz\left( \sum_{i \in r}  [AC]_{i*}\right) \leq \sum_{i \in r} \hat{z}_i \leq (1-\log^{-1} U)\nnz\left( \sum_{i \in r}  [AC]_{i*}\right). \label{eq:split}
\end{equation}

We will now argue that if we can create a split of the rows such that (\ref{enu:split1}) and (\ref{enu:split2}) hold, then when the splitting procedure terminates after $\log c + O(1)$ 
recursive levels, we have that for each pair of colors it is the case that $(A|S_i)(C|S'_j) < M / \log N$.
Consider the case where each split is done with the maximum positive error possible, i.e., on recursive level $q$ we have divided the $\nnz(AC)$ nonzeros into subproblems where each are of size at most  $\nnz(AC)(1/2 + 1/(2\log U))^q$. Remember that $c = \sqrt{\frac{\nnz(AC) \log N}{M}}$ is the number of colors. After $\log c + O(1)$ recursive levels we have subproblem size:

\begin{align}
\left( \frac{\nnz(AC)}{2} + \frac{\nnz(AC)}{2\log U} \right)^{\log c^2} &= \nnz(AC) 2^{-\log c^2} \left(1 + \frac{1}{\log U } \right)^{\log c^2} \nonumber \\
 & \leq \nnz(AC)2^{-\log c^2} e^{\frac{ \log c^2 }{\log U}} \label{eq:buckets2} \\
 & \leq \nnz(AC) O(1) / c^2\label{eq:buckets3} \\
 & = O( M / \log U ) \label{eq:buckets4}
\end{align}
The main observation to see that we get the right subproblem size as in (\ref{eq:buckets4}) is that for each recursion we decrease the output size by a factor $\Omega(c)$.
For (\ref{eq:buckets2}) we use $(1 + 1/x)^y \leq e^{y/x}$ and (\ref{eq:buckets3}) follows from $\nnz(AC) = N^{O(1)}$. The analysis for the case where each split is done with the maximum negative error possible is analogous and thus omitted.

We will now argue that with access to the $\hat{z}_i$ estimates as in (\ref{eq:zest}) we can always construct a split such that (\ref{enu:split1}) and (\ref{enu:split2}) hold.
Let partitions 1 and 2 be denoted $P_1$ and $P_2$ and $\hat{z} = \sum_i \hat{z}_i$ be the estimate of the total number of outputs for the current subproblem. Create $P_1$ by examining rows $[A]_{k*}$ one at a time.
If the estimated number of nonzeros of $P_1 \cup [A]_{k*}$ is less than $z/2$ then add $[A]_{k*}$ to $P_1$.
Otherwise perform dense-vector sparse-matrix multiplication $[A]_{k*}C$  using $\tilde{O}(\nnz(C)/B)$ I/Os \cite{brodal:spmv}and emit every nonzero of that product - this eliminates  the row vector $[A]_{k*}$ from matrix $A$ as all outputs generated by row $[A]_{k*}$ has now been emitted.
Because of (\ref{eq:split}) we have that the remaining rows of $A$ can now be placed in partition $P_2$ and the sum of their outputs will be at most $(1+\log^{-1} U)\nnz(AC)/2$. The procedure and analysis is equivalent for the case of columns.
From (\ref{eq:buckets4}) we had that even with splits of $\nnz(AC)(1/2 + \log (U)/2)$ nonzeros then the subproblem size is the desired $O(M / \log U)$ after all $\log c^2$ splits are done.

In terms of I/O complexity consider first the coloring of all rows in $A$. First we perform the size estimates of \Cref{lem:nnzest} in $\tilde{O}( N / B )$ such that we know where to split. Then we perform $c$ splits and each split also outputs the output entries for a specific row using dense-vector sparse-matrix multiplication, hence this split takes $c \nnz(C) = \tilde{O} \left( c N/B \right)$ I/Os. Finally for each of the $c$ sets of rows of $A$ we partition columns of $B$ in the same manner, first by invoking $c$ size estimations taking $\tilde{O} (N/B)$ due to the sum of the nonzeros in the $c$ subproblems being at most $N$. Then for each of the $c$ row sets we perform $c$ splits and output a column from $C$. This step takes time $\tilde{O}( c N / B)$ and hence in total we use
\(\tilde{O}( 3 c N / B + 2N/B) = \tilde{O} \left(\frac{N\sqrt{\nnz(AC)}}{B\sqrt{M}}\right).\)
\qed
\end{proof}

\subsection{I/O Complexity Analysis}\label{sec:ana}

Next, we will use \Cref{lem:color} for the algorithm that shows part (b) of \Cref{thm:main1}.

We summarize the steps taken and their cost in the external memory model.
\begin{proof}(\Cref{thm:main1}, part (b))
The algorithm first estimates $\nnz(AC)$ with parameters $\varepsilon = 1 / \log N$ and $\delta = 1 / U$ which by \Cref{lem:nnzest} uses $\tilde{O}( N / B )$ I/Os.
We then perform the coloring, outputting some entries of $AC$ and dividing the remaining entries into $c^2$ balanced sets for $c = \sqrt{\frac{\nnz(AC) \log U}{M}} +O(1)$.
By \Cref{lem:color} this uses $\tilde{O}\left(\frac{N\sqrt{\nnz(AC)}}{B\sqrt{M}}\right)$ I/Os.
Finally we invoke the compressed matrix multiplication algorithm from \Cref{lem:compressedmm}
on each subproblem.
This is possible since each subproblem has at most $\gamma M / \log U$ nonzeros entries in the output.
The total cost of this is $O(cN/B)$ I/Os,
since each nonzero entry in $A$ and $C$ is part of at most $c$ products,
and the cost of each product is simply the cost of scanning the input.
\qed
\end{proof}

%% file: compressedmm.tex
%!TEX root = main.tex

\subsection{Compressed matrix multiplication in the I/O model}\label{sec:compressedmm}

Let $\gamma > 0$ be a suitably small constant, and define $r = 4 \gamma M / \log U$.
We now describe an I/O-efficient algorithm for matrix products $AC$ with $\nnz(AC) \leq \gamma M / \log U = r/4$ nonzeros.
If $A$ is stored in column-major order and $C$ is stored in row-major order, the algorithm makes just a single scan over the matrices.

The algorithm is a variation of the one found in~\cite{Pagh:2013:CMM:2493252.2493254}, adapted to the semigroup I/O model.
Specifically, for some constant $\ell$ 
and $t=1,\dots,\ell\log U$ let $h_t,h'_t: [U]\rightarrow [r]$ be pairwise independent hash functions.
% Expected fraction corrupted is 1/4, probability of exceeding 1/2 is (e/4)^
The algorithm computes the following $\ell\log U$ polynomials of degree at most $2r$:
$$p_t(x) = \sum_{k=1}^U \left(\sum_{i=1}^U A_{i,k} x^{h_t(i)}\right) \left(\sum_{j=1}^U C_{k,j} x^{h'_t(j)}\right) \enspace .$$
It is not hard to see that the polynomial $\sum_{i=1}^U A_{i,k} x^{h_t(i)}$ can be computed in a single scan over column $i$ of $A$, using space $r$.
Similarly, we can compute the polynomial $\sum_{j=1}^U C_{k,j} x^{h'_t(j)}$ in space $r$ by scanning row $j$ of $C$.
As soon as both polynomials have been computed, we multiply them and add the result to the sum of products that will eventually be equal to $p_t(x)$.
This requires additional space $2r$, for a total space usage of $4r$.

Though a computationally less expensive approach is described in~\cite{Pagh:2013:CMM:2493252.2493254}, we present a simple method that (without using any I/Os) uses the polynomials $p_t(x)$, $t=1,\dots,\ell\log U$, to compute the set of entries in $AC$ with probability $1-U^{-3}$.
For every $i$ and $j$, to compute the value of $[AC]_{i,j}$ consider the coefficient of $x^{h_t(i)+h'_t(j)}$ in $p_t$, for $t=1,\dots,\ell\log U$.
For suitably chosen $c$, with probability $1-U^{-5}$ the value $[AC]_{i,j}$ is found in the majority of these coefficients.
The majority coefficient can be computed using just equality checks among semigroup elements~\cite{Boyer81}.
The analysis in~\cite{Pagh:2013:CMM:2493252.2493254} gives us, for a suitable choice of $\gamma$ and $\ell$, the following:
\begin{lemma}\label{lem:compressedmm}
	Suppose matrix $A$ is stored in column-major order, and $C$ is stored in row-major order.
	There exists an algorithm in the semiring I/O model augmented with equality test, and an absolute constant $\gamma > 0$, such that if
	$\nnz(AC) < \gamma M / \log U$
    the algorithm outputs the nonzero entries of $AC$ with probability $1-U^{-3}$, using just a single scan over the input matrices.
\end{lemma}

%% file: lowerbound.tex
%!TEX root = main.tex

\section{Lower bound}\label{sec:lowerbound}

Our lower bound generalizes that of Hong and Kung~\cite{Jia-Wei:1981:ICR:800076.802486} on the I/O complexity of dense matrix multiplication.
We extend the technique of~\cite{Jia-Wei:1981:ICR:800076.802486} while taking inspiration from lower bounds in~\cite{irony2004communication,pietracaprina2012space,ps13}.
The closest previous work is the lower bound in~\cite{ps13} on the I/O complexity of triangle enumeration in a graph, but new considerations are needed due to the fact that cancellations can occur.

Like the lower bound of Hong and Kung~\cite{Jia-Wei:1981:ICR:800076.802486}, our lower bound holds in a {\em semiring model\/} where:
\begin{itemize}
\item A memory block holds up to $B$ matrix entries (from the semiring), and internal memory can hold $M/B$ memory blocks.
\item Semiring elements can be multiplied and added, resulting in new semiring elements.
\item No other operations on semiring elements (e.g.~division, subtraction, or equality test) are allowed.
\end{itemize}

The model allows us to store sparse matrices by listing just non-zero matrix entries and their coordinates.
We note that our algorithm respects the constraints of the semiring model with one small exception:
It uses equality checks among semiring elements.

We require the algorithm to work for every semiring, and in particular over fields of infinite size such as the real numbers, and for arbitrary values of nonzero entries in $A$ and $C$. 
Since only addition and multiplication are allowed, we can consider each output value as a polynomial over nonzero entries of the input matrices. 
By the Schwartz-Zippel theorem~\cite[Theorem 7.2]{Motwani:1995:RA:211390} we know that two polynomials agree on all inputs if and only if they are identical.
Since we are working in the semiring model, the only way to get the term $A_{i,k} C_{k,j}$ in an output polynomial is to directly multiply these input entries.
That means that to compute an output entry $[AC]_{i,j}$ we need to compute a polynomial that is identical to the sum of elementary products $\sum_k A_{i,k} C_{k,j}$.
It is possible that the computation of this polynomial involves other nonzero terms, but these are required to cancel out.

We now argue that for every $N$ and $Z$ there exist matrices $A$ and $C$ with $\nnz(A)+\nnz(C) = \Theta(N)$ and $\nnz(AC) = \Theta(Z)$, for which every execution of an external memory algorithm in the semiring model must use $\Omega \left( \tfrac{N}{B} \min\left(\sqrt{\tfrac{Z}{M}},\tfrac{N}{M}\right) \right)$ I/Os.
Our lower bound holds for the {\em best possible execution\/}, i.e., even if the algorithm has been tailored to the structure of the input matrices.

The hard instance for the lower bound is a dense matrix product, which maximizes the number of elementary products. 
In particular, since we ignore constant factors we may assume that $\sqrt{Z}$ and $N/\sqrt{Z}$ are integers.
Let $A$ be a $(\sqrt{Z})$-by-$(N/\sqrt{Z})$ dense matrix, and let $C$ be a $(N/\sqrt{Z})$-by-$(\sqrt{Z})$ dense matrix.
Without loss of generality, every semiring element that is stored during the computation is either:
1) An input entry, or
2) Part of a sum that will eventually be emitted as the value of a unique nonzero element $[AC]_{i,j}$.

This is because these are the only values that can be used to compute an output entry (making use of the fact that additive and multiplicative inverses cannot be computed).
This implies that every output entry can be traced through the computation, and it is possible to pinpoint the time in the execution where an elementary product is computed and stored in internal memory.

We use the following lemma from~\cite{irony2004communication}:
\begin{lemma}\label{lem:irony}
In space $M$ the number of elementary products that can be computed and stored is at most $M^{3/2}$.
\end{lemma}

Following~\cite{ps13}, observe that any execution of an I/O efficient algorithm can be split into {\em phases\/} of $M/B$ I/Os.
By doubling the memory size to $2M$ we find an equivalent execution where every read I/O happens at the beginning of the phase (before any processing takes place), and every write I/O happens at the end of the phase.
For every phase we can therefore identify the set of at most $2M$ input and output entries that involved in the phase.

If all values needed for emitting a particular output entry are present in a phase there may not be any storage location that can be associated with it.
We first account for such {\em direct\/} outputs:
Each direct output requires two vectors of length $N/\sqrt{Z}$ to be stored in main memory.
In each phase we can store at most $M\sqrt{Z}/N$ such vectors, resulting in at most $M^2 Z/N^2$ output pairs.
So the number of phases needed to emit, say, $Z/2$ outputs would be at least $(N/M)^2$, using $N^2/(MB)$ I/Os.
This means that to output a substantial portion of $AC$ in this way we need at least this number of I/Os.

Next, we focus on output entries for which an elementary product is written to disk in some phase.
By Lemma~\ref{lem:irony} the number of elementary products computed and stored is at most $(2M)^{3/2}$.
If the total number of elementary products is $p$ then we need at least $p / (2M)^{3/2}$ phases of $M/B$ I/Os each.
Considering $Z/2$ output entries in our hard instance, these contain $N\sqrt{Z}/2$ elementary products.

Since $Z/2$ outputs are needed either in the direct or the indirect way, the number of I/Os needed becomes the minimum of the two lower bounds
we get Theorem~\ref{thm:lower}.

%% file: main.bbl
\begin{thebibliography}{10}

\bibitem{AV1988}
A.~Aggarwal and J.~S. Vitter.
\newblock The {I}nput/{O}utput complexity of sorting and related problems.
\newblock {\em Commun. ACM}, 31(9):1116--1127, 1988.

\bibitem{Alon:1995:COL:210332.210337}
N.~Alon, R.~Yuster, and U.~Zwick.
\newblock Color-coding.
\newblock {\em J. ACM}, 42(4):844--856, July 1995.

\bibitem{aloncount}
N.~Alon, R.~Yuster, and U.~Zwick.
\newblock Finding and counting given length cycles.
\newblock {\em Algorithmica}, 17(3):209--223, 1997.

\bibitem{Amossen:2010:BSE:1886521.1886554}
R.~R. Amossen, A.~Campagna, and R.~Pagh.
\newblock Better size estimation for sparse matrix products.
\newblock APPROX/RANDOM'10, pages 406--419, Berlin, Heidelberg, 2010.
  Springer-Verlag.

\bibitem{Amossen:2009:FJS:1514894.1514909}
R.~R. Amossen and R.~Pagh.
\newblock Faster join-projects and sparse matrix multiplications.
\newblock In {\em Proceedings of the 12th International Conference on Database
  Theory}, ICDT '09, pages 121--126, New York, NY, USA, 2009. ACM.

\bibitem{brodal:spmv}
M.~Bender, G.~Brodal, R.~Fagerberg, R.~Jacob, and E.~Vicari.
\newblock Optimal sparse matrix dense vector multiplication in the {I/O}-model.
\newblock {\em Theory of Computing Systems}, 47(4):934--962, 2010.

\bibitem{Boyer81}
R.~S. Boyer and J.~S. Moore.
\newblock {MJRTY} - {A} fast majority vote algorithm.
\newblock Technical Report AI81-32, The University of Texas at Austin,
  Department of Computer Sciences, Feb. 1 1981.

\bibitem{Cohen:1994:EST:1398518.1399012}
E.~Cohen.
\newblock Estimating the size of the transitive closure in linear time.
\newblock In {\em Proceedings of the 35th Annual Symposium on Foundations of
  Computer Science}, SFCS '94, pages 190--200, Washington, DC, USA, 1994. IEEE
  Computer Society.

\bibitem{cohen:1998}
E.~Cohen.
\newblock Structure prediction and computation of sparse matrix products.
\newblock {\em Journal of Combinatorial Optimization}, 2(4):307--332, 1998.

\bibitem{BRICS2002}
E.~D. Demaine.
\newblock Cache-oblivious algorithms and data structures.
\newblock In {\em Lecture Notes from the EEF Summer School on Massive Data
  Sets}. BRICS, University of Aarhus, Denmark, June 27--July 1 2002.

\bibitem{12823057}
P.~Flajolet and G.~N. Martin.
\newblock {Probabilistic Counting Algorithms for Data Base Applications}.
\newblock {\em Journal of Computer and System Sciences}.

\bibitem{greiner2010complexity}
G.~Greiner and R.~Jacob.
\newblock The i/o complexity of sparse matrix dense matrix multiplication.
\newblock In {\em LATIN 2010: Theoretical Informatics}, pages 143--156.
  Springer, 2010.

\bibitem{irony2004communication}
D.~Irony, S.~Toledo, and A.~Tiskin.
\newblock Communication lower bounds for distributed-memory matrix
  multiplication.
\newblock {\em Journal of Parallel and Distributed Computing},
  64(9):1017--1026, 2004.

\bibitem{Jia-Wei:1981:ICR:800076.802486}
H.~Jia-Wei and H.~T. Kung.
\newblock I/o complexity: The red-blue pebble game.
\newblock In {\em Proceedings of the Thirteenth Annual ACM Symposium on Theory
  of Computing}, STOC '81, pages 326--333, New York, NY, USA, 1981. ACM.

\bibitem{Kane:2010:OAD:1807085.1807094}
D.~M. Kane, J.~Nelson, and D.~P. Woodruff.
\newblock An optimal algorithm for the distinct elements problem.
\newblock In {\em Proceedings of the Twenty-ninth ACM SIGMOD-SIGACT-SIGART
  Symposium on Principles of Database Systems}, PODS '10, pages 41--52, New
  York, NY, USA, 2010. ACM.

\bibitem{mcgregorbook}
A.~McGregor.
\newblock {\em Algorithms for Signals, book draft}.
\newblock 2013.

\bibitem{Motwani:1995:RA:211390}
R.~Motwani and P.~Raghavan.
\newblock {\em Randomized Algorithms}.
\newblock Cambridge University Press, New York, NY, USA, 1995.

\bibitem{mvv89}
K.~Mulmuley, U.~Vazirani, and V.~Vazirani.
\newblock Matching is as easy as matrix inversion.
\newblock {\em Combinatorica}, 7(1):105--113, 1987.

\bibitem{Pagh:2013:CMM:2493252.2493254}
R.~Pagh.
\newblock Compressed matrix multiplication.
\newblock {\em ACM Trans. Comput. Theory}, 5(3):9:1--9:17, Aug. 2013.

\bibitem{ps13}
R.~Pagh and F.~Silvestri.
\newblock The input/output complexity of triangle enumeration.
\newblock {\em arXiv preprint arXiv:1312.0723}, 2013.

\bibitem{pietracaprina2012space}
A.~Pietracaprina, G.~Pucci, M.~Riondato, F.~Silvestri, and E.~Upfal.
\newblock Space-round tradeoffs for mapreduce computations.
\newblock In {\em Proceedings of the 26th ACM international conference on
  Supercomputing}, pages 235--244. ACM, 2012.

\bibitem{Rabin:1989:MMG:72302.72307}
M.~O. Rabin and V.~V. Vazirani.
\newblock Maximum matchings in general graphs through randomization.
\newblock {\em J. Algorithms}, 10(4):557--567, Dec. 1989.

\bibitem{strassen}
V.~Strassen.
\newblock Gaussian elimination is not optimal.
\newblock {\em Numerische Mathematik}, 13(4):354--356, 1969.

\bibitem{vandongen00}
S.~van Dongen.
\newblock {\em Graph Clustering by Flow Simulation}.
\newblock PhD thesis, University of Utrecht, 2000.

\bibitem{doi:10.1137/1.9781611973402.135}
R.~Williams and H.~Yu.
\newblock {\em Finding orthogonal vectors in discrete structures}, chapter 135,
  pages 1867--1877.

\bibitem{Williams:2012:MMF:2213977.2214056}
V.~V. Williams.
\newblock Multiplying matrices faster than coppersmith-winograd.
\newblock In {\em Proceedings of the Forty-fourth Annual ACM Symposium on
  Theory of Computing}, STOC '12, pages 887--898, New York, NY, USA, 2012. ACM.

\end{thebibliography}
